\documentclass[
a4paper,reqno]{amsart}
\usepackage[T1]{fontenc}
\usepackage[english]{babel}
\usepackage{amssymb,bbm,enumerate}
\usepackage{color}
\usepackage[linktocpage=true,colorlinks=true, linkcolor=blue, citecolor=red, urlcolor=green]{hyperref}
\usepackage{mathtools}
\usepackage{float}
\restylefloat{table}

\renewcommand{\phi}{\varphi}
\renewcommand{\ker}{\Ker}

\newcommand{\mc}[1]{\mathcal{#1}}
\newcommand{\mf}[1]{\mathfrak{#1}}
\newcommand{\mb}[1]{\mathbb{#1}}
\newcommand{\tint}{{\textstyle\int}}

\DeclareMathOperator{\Mat}{Mat}

\DeclareMathOperator{\ad}{ad}
\DeclareMathOperator{\im}{Im}
\DeclareMathOperator{\Der}{Der}
\DeclareMathOperator{\Ker}{Ker}
\DeclareMathOperator{\sdeg}{sdeg}
\DeclareMathOperator{\rank}{rank}

\newcommand{\ass}[1]{\stackrel{#1}{\longleftrightarrow}}

\theoremstyle{plain}
\newtheorem{theorem}{Theorem}[section]

\newtheorem{proposition}[theorem]{Proposition}

\theoremstyle{definition}
\newtheorem{definition}[theorem]{Definition}

\theoremstyle{remark}
\newtheorem{remark}[theorem]{Remark}

\numberwithin{equation}{section}

\definecolor{light}{gray}{.9}

\title{
Integrability of Dirac reduced bi-Hamiltonian equations
}

\author{Alberto De Sole, Victor G. Kac, Daniele Valeri}

\address{Dipartimento di Matematica, Sapienza Universit\`a di Roma,
P.le Aldo Moro 2, 00185 Rome, Italy.}
\email{desole@mat.uniroma1.it}

\address{Department of Mathematics, MIT,
77 Massachusetts Avenue, Cambridge, MA 02139, USA.}
\email{kac@math.mit.edu}

\address{SISSA, Via Bonomea 265, 34136 Trieste, Italy.}
\email{dvaleri@sissa.it}

\begin{document}

\pagestyle{plain}

\begin{abstract}
First, we give a brief review of the theory of the Lenard-Magri scheme for a non-local
bi-Poisson structure and of the theory of Dirac reduction.
These theories are used in the remainder of the paper to prove integrability of three
hierarchies of bi-Hamiltonian PDE's,
obtained by Dirac reduction from some generalized Drinfeld-Sokolov hierarchies.
\end{abstract}

\maketitle


\section{Introduction}\label{sec:0}

It has been demonstrated in a series of papers \cite{BDSK09,DSK13,DSKV13a,DSKV13b,DSKV13c}
that the framework of Poisson vertex algebras is extremely useful
for the theory of Hamiltonian PDE's.
For example, the theories of non-local Poisson structures \cite{DSK13},
and of the infinite dimensional Dirac reduction \cite{DSKV13c},
have been developed in this framework.
Moreover, this languages turned out to be very convenient
not only for the development of the general theory,
but also for the study of concrete bi-Hamiltonian systems,
like the generalized Drinfeld-Sokolov (DS) hierarchies,
considered in \cite{DSKV13a,DSKV13b}.
In these two papers we studied in more detail
three integrable bi-Hamiltonian hierarchies:
the homogeneous DS hierarchy, associated to a simple Lie algebra $\mf g$,
studied already in \cite{DS85},
and the generalized DS hierarchies attached to a minimal
and to a short nilpotent element of $\mf g$.
We also considered the Dirac reductions of the last two hierarchies
by elements of conformal weight 1.
In the case of a ``short'' hierarchy
we thus obtain Svinolupov's integrable hierarchy \cite{Svi91},
constructing thereby (non-local) bi-Poisson structures for them.
However, it is not at all clear (and probably false in general)
that the equations obtained by Dirac reduction from integrable bi-Hamiltonian equations
remain bi-Hamiltonian integrable.
We were able to prove this in \cite{DSKV13c}
for the reduced ``minimal'' hierarchy only in the first non-trivial case
of $\mf g=\mf{sl}_3$.

In the present paper, using the theory of singular degree of a rational matrix pseudodifferential
operator \cite{CDSK13c}, we prove integrability of the reduced ``minimal'' and ``short''
hierarchies for arbitrary $\mf g$.
Furthermore, considering Dirac reduction of the homogeneous DS hierarchy,
associated to a fixed regular element $s$ in a Cartan subalgebra $\mf h$ of $\mf g$,
we prove integrability of the following bi-Hamiltonian PDE, for all $a\in\mf h$:
\begin{equation}\label{eq1.1}
\frac{de_\alpha}{dt}=
\frac{\alpha(a)}{\alpha(s)}
e_\alpha^\prime
+\sum_{\beta\in\Delta\backslash\{-\alpha\}}
\frac{\beta(a)}{\beta(s)}
e_{-\beta}[e_\beta,e_\alpha]
\,\,,\,\,\,\,\alpha\in\Delta
\,,
\end{equation}
where $\Delta$ is the root system of $\mf g$
and $\{e_\alpha\}_{\alpha\in\Delta}$ are root vectors
such that $(e_\alpha|e_{-\alpha})=1$
with respect to an invariant non-degenerate bilinear form $(\cdot\,|\,\cdot)$ on $\mf g$.
Equation \eqref{eq1.1} is bi-Hamiltonian with respect to the following two compatible Poisson 
structures ($\alpha,\beta\in\Delta$):
\begin{equation}\label{eq1.2}
(H_0)_{\alpha,\beta}(\partial)=\delta_{\alpha,-\beta}\beta(s)
\,,
\end{equation}
and
\begin{equation}\label{eq1.3}
\begin{array}{l}
\vphantom{\Big)}
\displaystyle{
(H_1)_{\alpha,\beta}(\partial)=
[e_\beta,e_\alpha]-(\alpha|\beta)e_\alpha\partial^{-1}\circ e_\beta
\,\,\text{ for } \beta\neq-\alpha
\,,}\\
\vphantom{\Big)}
\displaystyle{
(H_1)_{\alpha,-\alpha}(\partial)=
\partial+(\alpha|\alpha)e_\alpha\partial^{-1}\circ e_{-\alpha}
\,.}
\end{array}
\end{equation}
The corresponding first two conserved Hamiltonian densities are
\begin{equation}\label{eq1.4}
h_0=a
\,\,,\,\,\,\,
h_1=\frac12\sum_{\alpha\in\Delta}
\frac{\alpha(a)}{\alpha(s)}
e_\alpha e_{-\alpha}\,.
\end{equation}

The proof of integrability in all cases is based on the Lenard-Magri scheme of integrability
for non-local bi-Poisson structures, developed in \cite{DSK13}.

\section{Non-local Poisson structures and Hamiltonian equations}\label{sec:2}

\subsection{Evolutionary vector fields, Frechet derivatives and variational derivatives}\label{sec:2.1}

Let $\mc V$ be the algebra of differential polynomials in $\ell$ variables:
$\mc V=\mb F [u_i^{(n)}\, |\, i \in I,n \in \mb Z_+]$, where $I=\{1,\dots,\ell\}$,
over a field $\mb F$ of characteristic $0$.
(In fact, most of the results hold in the generality of algebras 
of differential functions, as defined in \cite{DSK13}.)
It is a differential algebra with derivation defined by
$\partial(u_i^{(n)})=u_i^{(n+1)}$.
We also let $\mc K$ be the field of fractions of $\mc V$
(it is still a differential algebra).
We also denote by $\widetilde{\mc K}$
the \emph{linear closure} of $\mc K$,
which is the smallest differential field extension of $\mc K$
containing solutions to any linear differential equation
with coefficients in $\widetilde{\mc K}$,
and whose subfield of constants is $\overline{\mb F}$,
the algebraic closure of $\mb F$, see e.g. \cite{CDSK13b}.

For $P\in\mc V^\ell$ we have the associated \emph{evolutionary vector field}
$$
X_P=\sum_{i\in I,n\in\mb Z_+}(\partial^nP_i)\frac{\partial}{\partial u_i^{(n)}}\,\in\Der(\mc V)\,.
$$
This makes $\mc V^\ell$ into a Lie algebra, with Lie bracket
$[X_P,X_Q]=X_{[P,Q]}$, given by
$$
[P,Q]=X_P(Q)-X_Q(P)
=D_Q(\partial)P-D_P(\partial)Q
\,,
$$
where $D_P(\partial)$ and $D_Q(\partial)$ denote the Frechet derivatives of $P,Q\in\mc V^\ell$.

In general,
for $\theta=\big(\theta_\alpha\big)_{\alpha=1}^m\in\mc V^m$,
the \emph{Frechet derivative} 
$D_{\theta}(\partial)\in\Mat_{m\times\ell}\mc V[\partial]$ 
is defined by
\begin{equation}\label{20120126:eq2}
D_{\theta}(\partial)_{\alpha i}=
\sum_{n\in\mb Z_+}\frac{\partial \theta_\alpha}{\partial u_i^{(n)}}\partial^n
\,\,,\,\,\,\,
\alpha=1,\dots,m\,,\,\,i=1,\dots,\ell
\,.
\end{equation}
Its adjoint $D_{\theta}^*(\partial)\in\Mat_{\ell\times m}\mc V[\partial]$ is then given by
$$
D_{\theta}^*(\partial)_{i\alpha}=
\sum_{n\in\mb Z_+}(-\partial)^n
\frac{\partial \theta_\alpha}{\partial u_i^{(n)}}
\,\,,\,\,\,\,
\alpha=1,\dots,m\,,\,\,i=1,\dots,\ell
\,.
$$

For $f\in\mc V$ its \emph{variational derivative} is
$\frac{\delta f}{\delta u}=\left(\frac{\delta f}{\delta u_i}\right)_{i\in I}\in\mc V^{\oplus\ell}$,
where
$$
\frac{\delta f}{\delta u_i}
=\sum_{n\in\mb Z_+}(-\partial)^n\frac{\partial f}{\partial u_i^{(n)}}\,.
$$
Given an element $\xi\in\mc V^{\oplus\ell}$, the equation $\xi=\frac{\delta h}{\delta u}$ can be solved
for $h\in\mc V$ if and only if $D_\xi(\partial)$ is a self-adjoint operator:
$D_\xi(\partial)=D_\xi^*(\partial)$ (see e.g. \cite{BDSK09}).

\subsection{Rational matrix pseudodifferential operators}\label{sec:2.2}

Consider the skewfield $\mc K((\partial^{-1}))$ of pseudodifferential operators 
with coefficients in $\mc K$,
and the subalgebra $\mc V[\partial]$ of differential operators on $\mc V$.

The algebra $\mc V(\partial)$ of \emph{rational} pseudodifferential operators
consists of pseudodifferential operators $L(\partial)\in\mc V((\partial^{-1}))$
which admit a fractional decomposition
$L(\partial)=A(\partial)B(\partial)^{-1}$,
for some $A(\partial),B(\partial)\in\mc V[\partial]$, $B(\partial)\neq0$.
The algebra of \emph{rational matrix pseudodifferential operators}
is, by definition,  $\Mat_{\ell\times\ell}\mc V(\partial)$ \cite{CDSK12,CDSK13b}.

A matrix differential operator $B(\partial)\in\Mat_{\ell\times\ell}\mc V[\partial]$
is called \emph{non-degenerate}
if it is invertible in $\Mat_{\ell\times\ell}\mc K((\partial^{-1}))$.
Any matrix $H(\partial)\in\Mat_{\ell\times\ell}\mc V(\partial)$
can be written as a ratio of two matrix differential operators:
$H(\partial)=A(\partial) B^{-1}(\partial)$,
with $A(\partial),B(\partial)\in\Mat_{\ell\times\ell}\mc V[\partial]$,
and $B(\partial)$ non-degenerate.

\subsection{Singular degree of a rational matrix pseudodifferential operator}\label{sec:2.2c}

The \emph{Dieudonn\'e determinant} of
$A\in\Mat_{\ell\times\ell}\mc K((\partial^{-1}))$ is defined as follows.
If $A$ is degenerate, then $\det(A)=0$. Otherwise, $\det(A)$ is a pair
$$
\det(A)=(\textstyle{\det}_1(A),\deg(A))\in\mc K\times\mb Z\,,
$$
where $\det_1(A)$ and $\deg(A)$ are defined by the following conditions:
\begin{enumerate}[(i)]
\item
$\det_1(AB)=\det_1(A)\det_1(B)$ for all non-degenerate $A,B\in\Mat_{\ell\times\ell}\mc K((\partial^{-1}))$;
\item
$\deg(AB)=\deg(A)+\deg(B)$ for all non-degenerate $A,B\in\Mat_{\ell\times\ell}\mc K((\partial^{-1}))$;
\item
if $A$ is upper triangular, with diagonal entries $A_i=a_i\partial^{d_i}+$lower terms, $i=1,\dots,\ell$,
with $a_i\neq0$, then
$$
\textstyle{\det}_1(H)=\prod_{i=1}^\ell a_i
\,\,,\,\,\,\,
\deg(A)=\sum_{i=1}^\ell d_i\,.
$$
\end{enumerate}
For a non-degenerate $A\in\Mat_{\ell\times\ell}\mc K((\partial^{-1}))$, 
the integer $\deg(A)$ is called the \emph{degree} of $A$.
(It is a non-negative integer if $A$ is a matrix differential operator.)

Let $H\in\Mat_{\ell\times\ell}\mc V(\partial)$ be a rational matrix pseudodifferential operator.
The \emph{singular degree} of $H$, denoted $\sdeg(H)$ \cite{CDSK13c}, is, by definition, 
the minimal possible value of $\deg(B)$
among all fractional decomposition $H=AB^{-1}$,
with $A,B\in\Mat_{\ell\times\ell}\mc V[\partial]$,
and $B(\partial)$ non-degenerate.

Suppose that we have a rational expression for $H\in\Mat_{\ell\times\ell}\mc V(\partial)$
of the form
\begin{equation}\label{20130611:eq1}
H=\sum_{\alpha\in\mc A}
A^{\alpha}_1(B^\alpha_1)^{-1}\dots
A^{\alpha}_n(B^\alpha_n)^{-1}
\,,
\end{equation}
with $A^\alpha_i,B^\alpha_i\in\Mat_{\ell\times\ell}\mc K[\partial]$
and $B^\alpha_i$ non-degenerate, for all $i\in\mc I=\{1,\dots,n\},\alpha\in\mc A$
(a finite index set).
It is not hard to show that $\sdeg(H)\leq\sum_{\alpha\in\mc A}\sum_{i=1}^n\deg(B^\alpha_i)$,
\cite{CDSK13c}.
We say that the rational expression \eqref{20130611:eq1} is \emph{minimal} if
equality holds.
\begin{theorem}[{\cite[Cor.4.11]{CDSK13c}}]\label{20140116:thm1}
The rational expression \eqref{20130611:eq1}
is minimal if and only if
both the following systems of differential equations
in the variables 
$\{F^\alpha_i\}_{\alpha\in\mc A,i\in\{1,\dots,n\}}$
\begin{equation}\label{20140116:eq2}
\left\{\begin{array}{l}
B^\alpha_n F^\alpha_n=0
\,,\,\,\alpha\in\mc A \\
A^\alpha_iF^\alpha_i=B^\alpha_{i-1}F^\alpha_{i-1}
\,,\,\, 2\leq i\leq n,\,\alpha\in\mc A \\
\sum_{\alpha\in\mc A}A^\alpha_1F^\alpha_1=0
\end{array}\right.
\end{equation}
and
\begin{equation}\label{20140116:eq3}
\left\{\begin{array}{l}
{B^\alpha_1}^* F^\alpha_1=0
\,,\,\,\alpha\in\mc A \\
{A^\alpha_i}^*F^\alpha_{i-1}={B^\alpha_i}^*F^\alpha_i
\,,\,\, 2\leq i\leq n,\,\alpha\in\mc A \\
\sum_{\alpha\in\mc A}F^\alpha_n=0
\end{array}\right.
\end{equation}
have only the zero solution over the linear closure $\widetilde{\mc K}$ of $\mc K$.
\end{theorem}

\subsection{Association relation}\label{sec:2.2b}

Given $H(\partial)\in\Mat_{\ell\times\ell}\mc V(\partial)$,
we say that $\xi\in\mc V^{\oplus l}$ and $P\in\mc V^\ell$
are $H$-\emph{associated},
and denote it by
\begin{equation}\label{20130112:eq2}
\xi \ass{H}P\,,
\end{equation}
if there exist
a fractional decomposition $H=AB^{-1}$ 
with $A,B\in\Mat_{\ell\times\ell}\mc V[\partial]$ 
and $B$ non-degenerate,
and an element $F\in\mc K^\ell$,
such that $\xi=BF,\,P=AF$ \cite{DSK13}.

\begin{theorem}[{\cite[Thm4.12]{CDSK13c}}]\label{20140116:thm2}
Let \eqref{20130611:eq1} be a minimal rational expression for
$H(\partial)\in\Mat_{\ell\times\ell}\mc V(\partial)$.
Then, $\xi\ass{H}P$ if and only 
the system of differential equations
\begin{equation}\label{20140116:eq4}
\left\{\begin{array}{l}
B^\alpha_n F^\alpha_n=\xi
\,,\,\,\alpha\in\mc A \\
A^\alpha_iF^\alpha_i=B^\alpha_{i-1}F^\alpha_{i-1}
\,,\,\, 2\leq i\leq n,\,\alpha\in\mc A \\
\sum_{\alpha\in\mc A}A^\alpha_1F^\alpha_1=P
\end{array}\right.
\end{equation}
has a solution $\{F^\alpha_i\}_{\alpha\in\mc A,i\in\{1,\dots,n\}}$ over $\mc K$.
\end{theorem}

\subsection{Non-local Poisson structures}\label{sec:2.3}

To a matrix pseudodifferential operator
$H=\big(H_{ij}(\partial)\big)_{i,j\in I}\in\Mat_{\ell\times\ell}\mc V((\partial^{-1}))$
we associate a map, called $\lambda$-bracket,
$\{\cdot\,_\lambda\,\cdot\}_H:\,\mc V\times\mc V\to\mc V((\lambda^{-1}))$,
given by the following \emph{Master Formula} (see \cite{DSK13}):
\begin{equation}\label{20110922:eq1}
\{f_\lambda g\}_H
=
\sum_{\substack{i,j\in I \\ m,n\in\mb Z_+}} 
\frac{\partial g}{\partial u_j^{(n)}}
(\lambda+\partial)^n
H_{ji}(\lambda+\partial)
(-\lambda-\partial)^m
\frac{\partial f}{\partial u_i^{(m)}}
\,\in\mc V((\lambda^{-1}))
\,.
\end{equation}
In particular,
\begin{equation}\label{20130613:eq2}
H_{ji}(\partial)
=
{\{{u_i}_\partial{u_j}\}_H}_\to
\,.
\end{equation}
(The arrow means that we move $\partial$ to the right.)

The following facts are proved in \cite{BDSK09} and \cite{DSK13}.
For arbitrary $H$, the $\lambda$-bracket \eqref{20110922:eq1}
satisfies the following sesquilinearity conditions:
\begin{enumerate}[(i)]
\setcounter{enumi}{0}
\item
$\{\partial f_\lambda g\}=-\lambda \{f_\lambda g\}$,
\item
$\{f_\lambda \partial g\}=(\lambda+\partial)\{f_\lambda g\}$,
\end{enumerate}
and left and right Leibniz rules ($f,g,h\in\mc V$):
\begin{enumerate}[(i)]
\setcounter{enumi}{2}
\item
$\{f_\lambda gh\}=\{f_\lambda g\}h+\{f_\lambda h\}g$,
\item
$\{fg_\lambda h\}=\{f_{\lambda+\partial} h\}g+\{g_{\lambda+\partial} h\}f$.
\end{enumerate}
Here and further an expression $\{f_{\lambda+\partial}h\}_\to g$ is interpreted as follows:
if $\{f_{\lambda}h\}=\sum_{n=-\infty}^Nc_n\lambda^n$, 
then $\{f_{\lambda+\partial}h\}_\to g=\sum_{n=-\infty}^Nc_n(\lambda+\partial)^ng$,
where we expand $(\lambda+\partial)^n$ in non-negative powers of $\partial$.

Skewadjointness of $H$ is equivalent to the following skewsymmetry condition
\begin{enumerate}[(i)]
\setcounter{enumi}{4}
\item
$\{f_\lambda g\}=-\{g_{-\lambda-\partial} f\}$.
\end{enumerate}
The RHS of the skewsymmetry condition should be interpreted as follows:
we move $-\lambda-\partial$ to the left and
we expand its powers in non-negative powers of $\partial$,
acting on the coefficients on the $\lambda$-bracket.

Let
$\mc V_{\lambda,\mu}:=\mc V[[\lambda^{-1},\mu^{-1},(\lambda+\mu)^{-1}]][\lambda,\mu]$,
i.e. the quotient of the $\mb F[\lambda,\mu,\nu]$-module
$\mc V[[\lambda^{-1},\mu^{-1},\nu^{-1}]][\lambda,\mu,\nu]$
by the submodule
$(\nu-\lambda-\mu)\mc V[[\lambda^{-1},\mu^{-1},\nu^{-1}]][\lambda,\mu,\nu]$.
We have the natural embedding 
$\iota_{\mu,\lambda}:\,\mc V_{\lambda,\mu}\hookrightarrow V((\lambda^{-1}))((\mu^{-1}))$
defined by expanding the negative powers of $\nu=\lambda+\mu$
by geometric series in the domain $|\mu|>|\lambda|$.
In general, if $H$ is an arbitrary matrix pseudodifferentil operator,
we have
$\{f_\lambda\{g_\mu h\}\}\in\mc V((\lambda^{-1}))((\mu^{-1}))$ for all $f,g,h\in\mc V$.
If $H$ is a rational matrix pseudodifferential operator, we have
the following admissibility condition ($f,g,h\in\mc V$):
\begin{enumerate}[(i)]
\setcounter{enumi}{5}
\item
$\{f_\lambda\{g_\mu h\}\}\in\mc V_{\lambda,\mu}$,
\end{enumerate}
where we identify the space $\mc V_{\lambda,\mu}$
with its image in $\mc V((\lambda^{-1}))((\mu^{-1}))$ via the embedding $\iota_{\mu,\lambda}$.

\begin{definition}\label{20140117:def}
A \emph{non-local Poisson structure} on $\mc V$
is  a skewadjoint rational matrix pseudodifferential operator $H$ with coefficients in $\mc V$,
satisfying the following Jacobi identity ($f,g,h\in\mc V$):
\begin{enumerate}[(i)]
\setcounter{enumi}{6}
\item
$\{f_\lambda\{g_\mu h\}\}-\{g_\mu\{f_\lambda h\}\}=\{\{f_\lambda g\}_{\lambda+\mu} h\}$,
\end{enumerate}
where the equality is understood in the space $\mc V_{\lambda,\mu}$.
\end{definition}
(Note that, if skewsymmetry (v) and admissibility (vi) hold,
then all three terms of Jacobi identity lie in the image of $\mc V_{\lambda,\mu}$
via the appropriate embedding $\iota_{\mu,\lambda}$, $\iota_{\lambda,\mu}$ 
or $\iota_{\lambda+\mu,\lambda}$.)
Note that Jacobi identity (vii) holds for all $f,g,h\in\mc V$
if and only if it holds for any triple of generators $u_i,u_j,u_k$.

Two non-local Poisson structures $H_0,H_1\in\Mat_{\ell\times\ell}\mc V(\partial)$ on $\mc V$
are said to be \emph{compatible} if any their linear combination (or, equivalently, their sum)
is a non-local Poisson structure.
In this case we say that $(H_0,H_1)$ form a \emph{bi-Poisson structure} on $\mc V$.
\begin{definition}\label{pva}
A \emph{non-local Poisson vertex algebra}
is, by definition, a differential algebra $\mc V$ endowed with a $\lambda$-bracket
$\{\cdot\,_\lambda\,\cdot\}:\,\mc V\times\mc V\to\mc V((\lambda^{-1}))$
satisfying conditions (i)--(vii).
\end{definition}
We shall often drop the term ``non-local'',
so when we will refer to Poisson structures and $\lambda$-brackets
we will always mean \emph{non-local PVA}'s
and \emph{non-local} $\lambda$-\emph{brackets}.
(This, of course, includes the local case as well.)

\subsection{Hamiltonian equations and integrability}\label{sec:2.4}

Recall that we have a non-degenerate pairing 
$(\cdot\,|\,\cdot):\,\mc V^\ell\times\mc V^{\ell}\to\mc V/\partial\mc V$
given by $(P|\xi)=\tint P\cdot\xi$ (see e.g. \cite{BDSK09}).
Let $H\in\Mat_{\ell\times\ell}\mc V(\partial)$ be a non-local Poisson structure.
An evolution equation on the variables $u=\big(u_i\big)_{i\in I}$,
\begin{equation}\label{20120124:eq5}
\frac{du}{dt}
=P\,,
\end{equation}
is called \emph{Hamiltonian} with respect to the Poisson structure $H$
and the Hamiltonian functional $\tint h\in\mc V/\partial\mc V$
if (cf. Section \ref{sec:2.2b})
$$
\frac{\delta h}{\delta u}\ass{H}P
\,.
$$

Equation \eqref{20120124:eq5} is called \emph{bi-Hamiltonian}
if there are two compatible non-local Poisson structures $H_0$ and $H_1$,
and two local functionals $\tint h_0,\tint h_1\in\mc V/\partial\mc V$,
such that
\begin{equation}\label{20140117:eq1}
\frac{\delta h_0}{\delta u}\ass{H_1}P
\,\,\text{ and }\,\,
\frac{\delta h_1}{\delta u}\ass{H_0}P
\,.
\end{equation}

%
An \emph{integral of motion} for the Hamiltonian equation \eqref{20120124:eq5}
is a local functional $\tint f\in\mc V/\partial\mc V$
which is constant in time, i.e. such that $(P|\frac{\delta f}{\delta u})=0$.
The usual requirement for \emph{integrability}
is to have
sequences $\{\tint h_n\}_{n\in\mb Z_+}\subset\mc V/\partial\mc V$ 
and $\{P_n\}_{n\in\mb Z_+}\subset\mc V^\ell$,
starting with $\tint h_0=\tint h$ and $P_0=P$,
such that
\begin{enumerate}[(C1)]
\item
$\frac{\delta h_n}{\delta u}\ass{H}P_n$ for every $n\in\mb Z_+$,
\item
$[P_m,P_n]=0$ for all $m,n\in\mb Z_+$,
\item
$(P_m\,|\,\frac{\delta h_n}{\delta u})=0$ for all $m,n\in\mb Z_+$.
\item
The elements $P_n$ span an infinite dimensional subspace of $\mc V^\ell$.
\end{enumerate}
In this case, we have an \emph{integrable hierarchy} of Hamiltonian equations
$$
\frac{du}{dt_n} = P_n\,,\,\,n\in\mb Z_+\,.
$$
Elements $\tint h_n$'s are called \emph{higher Hamiltonians},
the $P_n$'s are called \emph{higher symmetries},
and the condition $(P_m\,|\,\frac{\delta h_n}{\delta u})=0$
says that $\tint h_m$ and $\tint h_n$ are \emph{in involution}.
Note that (C4) implies that elements $\frac{\delta h_n}{\delta u}$
span an infinite dimensional subspace of $\mc V^\ell$.
The converse holds provided that either $H_0$ or $H_1$ is non-degenerate.

Suppose we have a bi-Hamiltonian equation \eqref{20120124:eq5},
associated to the compatible Poisson structures $H_0,H_1$
and the Hamiltonian functionals $\tint h_0,\tint h_1$,
in the sense of equation \eqref{20140117:eq1}.
The \emph{Lenard-Magri scheme of integrability}
consists in finding
sequences $\{\tint h_n\}_{n\in\mb Z_+}\subset\mc V/\partial\mc V$ 
and $\{P_n\}_{n\in\mb Z_+}\subset\mc V^\ell$,
starting with $P_0=P$ and the given Hamiltonian functionals $\tint h_0,\tint h_1$,
satisfying the following recursive relations:
\begin{equation}\label{20130604:eq7}
\frac{\delta h_{n-1}}{\delta u}\ass{H_1}P_n
\,\,,\,\,\,\,
\frac{\delta h_n}{\delta u}\ass{H_0}P_n
\,\,\,\,
\text{ for all } n\in\mb Z_+
\,.
\end{equation}
In this case,
we have the corresponding bi-Hamiltonian hierarchy
\begin{equation}\label{20130604:eq6}
\frac{du}{dt_n}=P_n\,\in\mc V^\ell
\,\,,\,\,\,\,
n\in\mb Z_+
\,,
\end{equation}
all Hamiltonian functionals $\tint h_n,\,n\geq-1$,
are integrals of motion for all equations of the hierarchy,
and they are in involution with respect to both Poisson structures $H_0$ and $H_1$,
and all commutators $[P_m,P_n]$ are zero, provided that one of the Poisson 
structures $H_0$ or $H_1$ is local (see \cite[Sec.7.4]{DSK13}).
Hence, in this situation \eqref{20130604:eq6} is an integrable hierarchy
of compatible evolution equations,
provided that condition (C4) holds.

\section{Dirac reduction for (non-local) Poisson structures and Hamiltonian equations}\label{sec:3}

\subsection{Dirac reduction of a Poisson structure}\label{sec:3.1}

Let $H(\partial)\in\Mat_{\ell\times\ell}\mc V(\partial)$
be a Poisson structure on $\mc V$.
Let $\{\cdot\,_\lambda\,\cdot\}_H$ be the corresponding PVA $\lambda$-bracket on $\mc V$
given by the Master Formula \eqref{20110922:eq1}.
Let $\theta_1,\dots,\theta_m$ be some elements of $\mc V$,
and let $\mc I=\langle\theta_1,\dots,\theta_m\rangle_{\mc V}\subset\mc V$ 
be the differential ideal generated by them.
Consider the following rational matrix pseudodifferential operator
\begin{equation}\label{20130529:eq1}
C(\partial)=D_\theta(\partial)\circ H(\partial)\circ D_\theta^*(\partial)\,
\in\Mat_{m\times m}\mc V(\partial)\,,
\end{equation}
where $D_\theta(\partial)$ is the $m\times\ell$ matrix differential operator
of Frechet derivatives of the elements $\theta_i$'s:
\begin{equation}\label{20130529:eq2}
{D_\theta(\partial)}_{\alpha,i}
=\sum_{n\in\mb Z_+}\frac{\partial\theta_\alpha}{\partial u_i^{(n)}}\partial^n
\,\,,\,\,\,\,
\alpha=1,\dots,m,\,i=1,\dots,\ell\,,
\end{equation}
and $D_\theta^*(\partial)\in\Mat_{\ell\times m}\mc V[\partial]$ is its adjoint.
Recalling the Master Formula \eqref{20110922:eq1},
we get that $C(\partial)$ has matrix elements with symbol
\begin{equation}\label{C}
C_{\alpha\beta}(\lambda)=\{\theta_{\beta}{}_{\lambda}\theta_{\alpha}\}_H\,.
\end{equation}
Note also that, by the skewadjointness of $H$,
the corresponding $\lambda$-bracket $\{\cdot\,_\lambda\,\cdot\}_H$
is skewsymmetric, hence $C(\partial)$ is a skewadjoint pseudodifferential operator.

We shall assume that the matrix $C(\partial)$ in \eqref{20130529:eq1}
is invertible in $\Mat_{m\times m}\mc V((\partial^{-1}))$,
and we denote its inverse by
$C^{-1}(\partial)=\big((C^{-1})_{\alpha\beta}(\partial)\big)_{\alpha,\beta=1}^m
\in\Mat_{m\times m}\mc V((\partial^{-1}))$.
\begin{definition}\label{20130529:def}
The \emph{Dirac modification} of the Poisson structure $H\in\Mat_{\ell\times\ell}\mc V(\partial)$
by the \emph{constraints} $\theta_1,\dots,\theta_m$
is the following skewadjoint $\ell\times\ell$ matrix pseudodifferential operator:
\begin{equation}\label{20130529:eq3}
H^D(\partial)
=
H(\partial)+B(\partial)\circ C^{-1}(\partial)\circ B^*(\partial)\,,
\end{equation}
where $B(\partial)=H(\partial)\circ D_\theta^*(\partial)
\in\Mat_{\ell\times m}\mc V(\partial)$.
\end{definition}
The matrix pseudodifferential operator $H^D(\partial)$ is skewadjoint and rational.
The corresponding $\lambda$-bracket, given by the Master Formula \eqref{20110922:eq1}, 
is (cf. \cite{DSKV13c})
\begin{equation}\label{dirac}
\{f_{\lambda}g\}_H^D
=\{f_{\lambda}g\}_H
-\sum_{\alpha,\beta=1}^m
{\{{\theta_{\beta}}_{\lambda+\partial}g\}_H}_{\to}
(C^{-1})_{\beta\alpha}(\lambda+\partial)
{\{f_{\lambda}\theta_{\alpha}\}_H}\,.
\end{equation}
The following result is a special case of \cite[Thm.2.2]{DSKV13c}:
\begin{theorem}[{}]\label{prop:dirac}
\begin{enumerate}[(a)]
\item
The Dirac modified $\lambda$-bracket \eqref{dirac} satisfies the Jacobi identity (vii).
Consequently, the Dirac modification $H^D(\partial)$ is a non-local Poisson structure on $\mc V$.
\item
All the elements $\theta_i,\,i=1,\dots,m$, are central 
with respect to the Dirac modified $\lambda$-bracket, i.e.:
$$
\{f_\lambda\theta_i\}^D_H=\{{\theta_i}_\lambda f\}^D_H=0
$$
for all $i=1,\dots,m$ and $f\in\mc V$.
\item
The differential ideal $\mc I=\langle\theta_1,\dots,\theta_m\rangle_{\mc V}\subset\mc V$,
generated by $\theta_1,\dots,\theta_m$,
is an ideal with respect to the Dirac modified $\lambda$-bracket $\{\cdot\,_\lambda\,\cdot\}^D$,
namely:
$$
\{\mc I\,_\lambda\,\mc V\}^D_H
\,,\,\,
\{\mc V\,_\lambda\,\mc I\}^D_H\,
\subset\mc I((\lambda^{-1}))
\,.
$$
Hence, the quotient space $\mc V/\mc I$ is a PVA,
with $\lambda$-bracket induced by $\{\cdot\,_\lambda\,\cdot\}^D$,
which we call the \emph{Dirac reduction} of $\mc V$ 
by the constraints $\theta_1,\dots,\theta_m$.
\end{enumerate}
\end{theorem}
\begin{remark}\label{20140117:rem}
If the constraints $\theta_i$'s are some generators of 
the algebra of differential polynomials $\mc V$,
then the quotient $\mc V/\mc I$ is still an algebra of differential polynomials
(in the remaining generators),
and we have the induced Poisson structure $\overline{H}^D$
on this quotient
(corresponding to the PVA $\lambda$-bracket of Theorem \ref{prop:dirac}(c)).
\end{remark}

\subsection{Dirac reduction of a bi-Poisson structure}\label{sec:3.2}

Let $(H_0,H_1)$ be a bi-Poisson structure on $\mc V$.
Let $\theta_1,\dots,\theta_m\in\mc V$ be central elements for $H_0$.
Suppose that the matrix pseudodifferential operator (cf. \eqref{20130529:eq1})
$C(\partial)=D_\theta(\partial)\circ H_1(\partial)\circ D_\theta^*(\partial)$
is invertible.
Then we can consider the Dirac modified Poisson structure $H_1^D$ (cf. \eqref{20130529:eq3}),
and the corresponding $\lambda$-bracket $\{\cdot\,_\lambda\,\cdot\}_1^D$ (cf. \eqref{dirac}),
and we have the following result:
\begin{theorem}[{\cite[Thm.2.3]{DSKV13c}}]\label{20140117:thm}
\begin{enumerate}[(a)]
\item
The matrices $H_0$ and $H_1^D$ form a compatible pair of Poisson structures on $\mc V$.
\item
The differential algebra ideal
$\mc I=\langle\theta_1,\dots,\theta_m\rangle_{\mc V}$
is a PVA ideal for both the $\lambda$-brackets 
$\{\cdot\,_\lambda\,\cdot\}_0$
and $\{\cdot\,_\lambda\,\cdot\}_1^D$,
and we have the induced compatible PVA $\lambda$-brackets on $\mc V/\mc I$.
\end{enumerate}
\end{theorem}

\subsection{Reduction of a bi-Hamiltonian hierarchy}\label{sec:3.3}

Let $(H_0,H_1)$ be a \emph{local} bi-Poisson structure
(i.e. consisting of matrix differential operators).
Suppose that we have a bi-Hamiltonian hierarchy 
$\frac{du}{dt_n}=P_n\in\mc V^\ell$, $n\in\mb Z_+$,
with respect to $(H_0,H_1)$,
and let $\tint h_n\in\mc V/\partial\mc V$
be a sequence of integrals of motion
satisfying the Lenard-Magri recursive condition \eqref{20130604:eq7}.
Let $\theta_1,\dots,\theta_m\in\mc V$ be central elements for $H_0$.
Assume that the matrix
$C(\partial)=D_\theta(\partial)\circ H_1(\partial)\circ D_\theta^*(\partial)\in\Mat_{m\times m}\mc V[\partial]$
is invertible in $\Mat_{m\times m}\mc V((\partial^{-1}))$.
Then, by Theorem \ref{20140117:thm},
$H_1^D=H_1+B(\partial)C^{-1}(\partial)B^*(\partial)$,
where $B(\partial)=H_1(\partial)\circ D_\theta^*(\partial)$,
is a (non-local) Poisson structure on $\mc V$ compatible to $H_0$.
Moroever, we have the following result:
\begin{proposition}\label{20131203:prop2}
Suppose that $\Ker B(\partial)$ and $\Ker C(\partial)$
have zero intersection over the linear closure $\widetilde{\mc K}$ of $\mc K$.
Then we have the Lenard-Magri recursive relations
\begin{equation}\label{20131203:eq1}
\frac{\delta h_{n-1}}{\delta u}\ass{H_1^D}P_n
\,\,,\,\,\,\,
\frac{\delta h_n}{\delta u}\ass{H_0}P_n
\,\,\,\,
\text{ for all } n\in\mb Z_+
\,.
\end{equation}
\end{proposition}
\begin{proof}
According to Theorem \ref{20140116:thm1},
the condition $\Ker B(\partial)$ and $\Ker C(\partial)$
have zero intersection
is equivalent to saying that 
$H_1^D=H_1+B(\partial)C^{-1}(\partial)B^*(\partial)$
is a minimal rational expression for $H_1^D$.
By assumption, we have 
$P_n=H_0\frac{\delta h_n}{\delta u}=H_1\frac{\delta h_{n-1}}{\delta u}$.
By Theorem \ref{20140116:thm2}
the association relation $\frac{\delta h_{n-1}}{\delta u}\ass{H_1^D}P_n$
holds if there exists $F_n\in\mc K^m$ such that
$$
B^*(\partial)\frac{\delta h_{n-1}}{\delta u}
=C(\partial)F_n
\qquad\text{and}\qquad
B(\partial)F_n=0\,.
$$
Note that $B^*(\partial)\frac{\delta h_{n-1}}{\delta u}
=-D_{\theta}(\partial)H_1(\partial)\frac{\delta h_{n-1}}{\delta u}
=-D_{\theta}(\partial)P_n$.
Since the elements $\theta_\alpha$'s are central for the Poisson structure $H_0$,
they are constant densities for the Hamiltonian equations \eqref{20130604:eq6}
(see \cite[Lem.5.2(b)]{DSKV13c}). Thus we have $D_{\theta}(\partial)P_n=0$,
for every $n\in\mb Z_+$.
Therefore we can choose $F_n$ to be the zero vector in $\mc K^\ell$, for every $n\in\mb Z_+$.
\end{proof}

\section{Dirac reduced homogeneous DS hierarchy}

First we review the construction of the homogeneous Drinfeld-Sokolov hierarchy,
following \cite{DSKV13a}.

Let $\mf g$ be a simple finite-dimensional Lie algebra. 
Fix a non-degenerate symmetric invariant bilinear form $(\cdot\,|\,\cdot)$ on $\mf g$,
and a regular semisimple element $s\in\mf g$.
We have the direct sum
decomposition $\mf g=\mf h\oplus\mf h^{\perp}$, 
where $\mf h=\Ker(\ad s)$ is a Cartan subalgebra,
and $\mf h^{\perp}=\im(\ad s)$ is its orthogonal complement with respect to the
bilinear form $(\cdot\,|\,\cdot)$,
and it is the direct sum of root spaces.

Let $\mc V=S(\mb F[\partial]\mf g)$,
the algebra of differential polynomials in a basis of $\mf g$.
We define a $\lambda$-bracket on $\mc V$ by
\begin{equation}\label{20131204:eq1}
\{a_{\lambda}b\}_z=[a,b]+(a| b)\lambda+z(s|[a,b])\,,
\end{equation}
for $a,b\in\mf g$, and we extend it to a $\lambda$-bracket on
$\mc V$ by \eqref{20110922:eq1}, thus obtaining a PVA structure.
In equation \eqref{20131204:eq1} $z\in\mb F$ is a parameter.

Let $\ell=\rank(\mf g)$ be the rank of $\mf g$, and let $\Delta$ be the set of roots of $\mf g$.
Choose a basis $\mf g$ as follows:
$\mc B=\{x_i\}_{i=1}^\ell\cup\{e_\alpha\}_{\alpha\in\Delta}$,
union of an orthonormal basis of $\mf h$ and a collection of root vectors
such that $(e_\alpha|e_{-\alpha})=1$.
Hence $\mc V=\mb C[x_i^{(n)},e_{\alpha}^{(n)}\mid i\in\{1,\dots,\ell\},\alpha\in\Delta,n\in\mb Z_+]$ is 
the algebra of differential polynomials generated by the elements of the basis $\mc B$.
Equation \eqref{20131204:eq1} defines a (local) bi-Poisson structure $(H_0,H_1)$ on $\mc V$ 
given by ($i,j\in\{1,\dots,\ell\},\alpha\in\Delta$)
\begin{equation}\label{20131204:eq2}
\left\{
\begin{array}{l}
(H_0)_{ij}(\partial)=0\\
(H_0)_{\alpha i}(\partial)=0\\
(H_0)_{\alpha\beta}(\partial)=\delta_{\alpha,-\beta}\beta(s)
\end{array}\right.
\quad\text{and}\quad
\left\{
\begin{array}{l}
(H_1)_{ij}(\partial)=\delta_{ij}\partial\\
(H_1)_{\alpha i}(\partial)=\alpha(x_i)e_\alpha\\
(H_1)_{\alpha\beta}(\partial)=[e_\beta,e_\alpha]+\delta_{\alpha,-\beta}\partial\,.
\end{array}
\right.
\end{equation}

It is proved in
\cite{DSKV13a} that, for any $a\in\mf h$, we have an infinite sequence of local functionals
$\{\tint h_n\}_{n\geq-1}\subset\mc V/\partial\mc V$,
satisfying the Lenard-Magri recursive relations \eqref{20130604:eq7},
and $h_{-1}=0$, $h_0=a$.
The next two functionals of the sequence have densities
$$
h_1=\frac12\sum_{\alpha\in\Delta}
\frac{\alpha(a)}{\alpha(s)}
e_\alpha e_{-\alpha}\,,
$$ 
and
$$
\begin{array}{l}
\displaystyle{
h_2=\frac12\sum_{\alpha\in\Delta}\frac{\alpha(a)}{\alpha(s)^2}e_{-\alpha}e_{\alpha}^\prime
+\frac12\sum_{\alpha\in\Delta}\frac{\alpha(a)}{\alpha(s)^2}e_{\alpha}e_{-\alpha}[e_{-\alpha},e_{\alpha}]
}\\
\displaystyle{
+\frac13\sum_{\substack{\alpha,\beta\in\Delta\\\alpha\neq\beta}}\frac{\alpha(a)}{\alpha(s)\beta(s)}
e_{\alpha}e_{-\beta}[e_{-\alpha},e_{\beta}]\,.
}
\end{array}
$$
The corresponding Hamiltonian equations \eqref{20130604:eq6}
are: $\frac{dx_i}{dt_0}=\frac{dx_i}{dt_1}=\frac{dx_i}{dt_2}=0$,
for $i=1,\dots,\ell$, and, for $\alpha\in\Delta$,
\begin{equation}\label{hier:DS}
\begin{array}{l}
\displaystyle{
\frac{de_\alpha}{dt_0}
=\alpha(a)e_\alpha
\,,\qquad
\frac{de_\alpha}{dt_1}
=\frac{\alpha(a)}{\alpha(s)}e_{\alpha}^\prime
+\sum_{\beta\in\Delta}
\frac{\beta(a)}{\beta(s)}
e_{-\beta}[e_{\beta},e_\alpha]\,,
}\\
\displaystyle{
\frac{de_\alpha}{dt_2}
=\frac{\alpha(a)}{\alpha(s)^2}e_{\alpha}''
+\frac{\alpha(a)}{\alpha(s)^2}(e_{\alpha}[e_{-\alpha},e_{\alpha}])'
+\sum_{\beta\in\Delta}\frac{\beta(a)}{\beta(s)^2}
e_\beta'[e_{-\beta},e_{\alpha}]
}
\\
\displaystyle{
+\frac13\sum_{\beta\in\Delta\backslash\{\alpha\}}
\left(\frac{\alpha(a)+\beta(a)}{\alpha(s)\beta(s)}
+\frac{\beta(a)}{\beta(s)(\alpha(s)-\beta(s))}\right)
(e_\beta[e_{-\beta},e_\alpha])'
}
\\
\displaystyle{
+\sum_{\beta\in\Delta}\frac{\beta(a)}{\beta(s)^2}\left(
e_{\beta}[e_{-\beta},e_{\beta}][e_{-\beta},e_{\alpha}]
-\frac12(\alpha|\beta)e_{\alpha}e_{\beta}e_{-\beta}
\right)
}
\\
\displaystyle{
+\frac13\sum_{\substack{\beta,\gamma\in\Delta\\\beta\neq\gamma}}
\frac{2\beta(a)\gamma(s)-\gamma(a)\beta(s)}{\beta(s)\gamma(s)(\beta(s)-\gamma(s))}
e_\beta e_{-\gamma}[[e_{-\beta},e_{\gamma}],e_{\alpha}]\,.
}
\end{array}
\end{equation}
It follows from \cite[Rem.2.7]{DSKV13a} 
that, for every $\alpha\in\Delta$,
\begin{equation}\label{20140119:eq1}
\frac{\delta h_n}{\delta e_\alpha}
=(-1)^{n+1}\frac{\alpha(a)}{\alpha(s)^{n}}e_{-\alpha}^{(n-1)}
+\text{higher polynomial order terms}\,.
\end{equation}
In particular, the elements $\frac{\delta h_n}{\delta u}$ are linearly independent.

Let $\mc I$ be the differential ideal generated by
the variables $x_i$, $i=1,\ldots,\ell$. Clearly, as differential algebras,
$$
\mc V/\mc I\simeq
\overline{\mc V}=
\mb F[e_\alpha^{(n)}\mid \alpha\in\Delta,n\in\mb Z_+]\,.
$$
Note that, by
equation \eqref{20131204:eq2},
the elements $x_i$, $i=1,\dots,\ell$, are central for the Poisson structure $H_0$.
Consider the Dirac modification $H_1^D$ 
of $H_1$ by the constraints $\{x_i\}_{i=1}^\ell$,
defined by the equation \eqref{20130529:eq3}:
$$
H_1^D(\partial)=H_1(\partial)+B(\partial)\circ C^{-1}(\partial)\circ B^*(\partial)\,,
$$
where the matrices $B(\partial)\in\Mat_{(\ell+|\Delta|)\times\ell}\mc V[\partial]$ 
and $C(\partial)\in\Mat_{\ell\times\ell}\mc V[\partial]$
have entries
\begin{equation}\label{20131204:eq3}
\begin{array}{ll}
B_{ij}(\partial)=C_{ij}(\partial)=(H_1)_{ij}(\partial)=\delta_{ij}\partial\,,
&
i,j=1,\dots,\ell\\
B_{\alpha i}(\partial)=
(H_1)_{\alpha i}(\partial)
=\alpha(x_i)e_{\alpha}\,,
&
i=1,\dots,\ell,\alpha\in\Delta\,.
\end{array}
\end{equation}
By Theorem \ref{20140117:thm}, we have a bi-Poisson structure
$(H_0,H_1^D)$ on $\mc V$,
and the induced bi-Poisson structure $(\overline{H}_0,\overline{H}_1^D)$ on $\overline{\mc V}$.
It is given by ($\alpha,\beta\in\Delta$)
$$
\begin{array}{l}
\displaystyle{
(\overline{H}_0)_{\alpha\beta}(\lambda)=\delta_{\alpha,-\beta}\beta(s)
}\\
\displaystyle{
(\overline{H}_1^D)_{\alpha,-\alpha}(\lambda)
=\partial+(\alpha|\alpha)e_\alpha\partial^{-1}\circ e_{-\alpha}\,,
}\\
\displaystyle{
(\overline{H}_1^D)_{\alpha\beta}(\lambda)
=[e_\beta,e_\alpha]
-(\alpha|\beta)e_\alpha\partial^{-1}\circ e_{\beta}
\,,\text{ for }\alpha\neq-\beta\,.
}
\end{array}
$$
\begin{proposition}
The Lenard-Magri recursive relations \eqref{20131203:eq1} hold for the bi-Poisson structure
$(H_1^D,H_0)$. Hence, we get an induced bi-Hamiltonian hierarchy in $\overline{\mc V}$.
\end{proposition}
\begin{proof}
By Proposition \ref{20131203:prop2} it suffices to show that $\ker B(\partial)=0$ over 
the linear closure $\widetilde{\mc K}$ of $\mc K$.
Let $F=(F_i)_{i=1}^\ell\in\widetilde{\mc K}$ be an element of the kernel of $B(\partial)$.
We have, for $\alpha\in\Delta$,
$$
\left(B(\partial)F\right)_\alpha
=\left(\alpha(x_1)F_1+\dots+\alpha(x_\ell)F_\ell\right)e_\alpha\,.
$$
Since $\Delta$ spans $\mf h^*$, it follows that $F=0$.
To conclude, we just observe that, by equation \eqref{20140119:eq1},
the images of elements $\frac{\delta h_n}{\delta u}$ in $\overline{\mc V}^{|\Delta|}$
are linearly independent.
Since $\overline{H}_0$ is an invertible constant matrix,
the images of the elements $P_n$ in $\overline{\mc V}^{|\Delta|}$
are linearly independent as well.
\end{proof}
\begin{remark}\label{20140122:rem}
It follows from \cite[Prop.2.10]{BDSK09}
and the definition of the Lie bracket between Hamiltonian functionals given in \cite[Eq.(7.8)]{DSK13}
(using the fact that $H_0$ is local)
that all the $\tint h_n$'s obtained by taking all possible $a\in\mf h$
are in involution.
\end{remark}
The first equations of the reduced bi-Hamiltonian hierarchy are ($\alpha\in\Delta$)
$$
\begin{array}{l}
\displaystyle{
\frac{de_\alpha}{dt_0}=\alpha(a)e_\alpha
\,,\qquad
\frac{de_\alpha}{dt_1}=
\frac{\alpha(a)}{\alpha(s)}
e_{\alpha}^\prime
+\sum_{\beta\in\Delta\backslash\{-\alpha\}}
\frac{\beta(a)}{\beta(s)}
e_{-\beta}[e_{\beta},e_\alpha]\,,
}
\\
\displaystyle{
\frac{de_\alpha}{dt_2}
=\frac{\alpha(a)}{\alpha(s)^2}e_{\alpha}''
+\sum_{\beta\in\Delta\backslash\{\alpha\}}\frac{\beta(a)}{\beta(s)^2}
e_\beta'[e_{-\beta},e_{\alpha}]
}
\\
\displaystyle{
+\frac13\sum_{\beta\in\Delta\backslash\{\alpha\}}
\left(\frac{\alpha(a)+\beta(a)}{\alpha(s)\beta(s)}
+\frac{\beta(a)}{\beta(s)(\alpha(s)-\beta(s))}\right)
(e_\beta[e_{-\beta},e_\alpha])'
}
\\
\displaystyle{
-\frac12\sum_{\beta\in\Delta}\frac{\beta(a)}{\beta(s)^2}
(\alpha|\beta)e_{\alpha}e_{\beta}e_{-\beta}
+\frac13\!\!\!\sum_{\substack{\beta,\gamma\in\Delta\\\beta\neq\gamma,\gamma+\alpha}}\!\!\!
\frac{2\beta(a)\gamma(s)-\gamma(a)\beta(s)}{\beta(s)\gamma(s)(\beta(s)-\gamma(s))}
e_\beta e_{-\gamma}[[e_{-\beta},e_{\gamma}],e_{\alpha}]\,.
}
\end{array}
$$
\begin{remark}
For $\mf g=\mf{sl}_2$ we have $\Delta=\{\alpha,-\alpha\}$.
Letting $s=a=[e_\alpha,e_{-\alpha}]$,
the first non trivial equation of the reduced bi-Hamiltonian 
hierarchy is
$$
\left\{\begin{array}{l}
\displaystyle{
\frac{de_{\alpha}}{dt_2}=\frac12e_\alpha''-e_\alpha^2e_{-\alpha}
}
\\
\displaystyle{
\frac{de_{-\alpha}}{dt_2}=-\frac12e_{-\alpha}''+e_{\alpha}e_{-\alpha}^2\,.
}
\end{array}\right.
$$
Hence, the reduced DS homogeneous hierarchy for the Lie algebra
$\mf g=\mf{sl}_2$ coincides with the \emph{NLS hierarchy} (AKNS),
and $(\overline{H}_0,\overline{H}_1^D)$ coincides with its well-known bi-Poisson
structure. 
\end{remark}

\section{Dirac reduced minimal DS hierarchy}
\label{sec:4.3}

We recall here the construction of the classical $\mc W$-algebra associated to a 
minimal nilpotent element following \cite{DSKV13b}.

Let $\mf g$ be a simple Lie algebra with a non-degenerate symmetric invariant
bilinear form $(\cdot\,|\,\cdot)$, and let $f\in\mf g$ be a minimal nilpotent 
element, that is a lowest root vector of $\mf g$. Let $\{f,h=2x,e\}\subset\mf g$ be 
an $\mf{sl}_2$-triple.
The $\ad x$-eigenspace decomposition is
$$
\mf g=\mb Ff\oplus\mf g_{-\frac12}\oplus\mf g_{0}\oplus\mf g_{\frac12}\oplus\mb Fe\,.
$$
Note that $(x|a)=0$ for all $a\in\mf g_0^f$.
Hence, the subalgebra $\mf g_0\subset\mf g$ admits the orthogonal decomposition
$\mf g_0=\mf g_0^f\oplus\mb Fx$.
For $a\in\mf g_0$ we denote
by $a^\sharp$ its projection to $\mf g_0^f$.

We fix a basis of $\mf g^f=\mf g_0^f\oplus\mf g_{-\frac12}\oplus\mb Ff$ as follows.
Let $\{a_i\}_{i\in J_0^f}\subset\mf g_0^f$ be an orthonormal basis of $\mf g_0^f$ 
with respect to $(\cdot\,|\,\cdot)$.
Let also $\{v_k\}_{k\in J_{-\frac12}}\subset\mf g_\frac12$ be a basis of
$\mf g_{\frac12}$ and let $\{v^k\}_{k\in J_{-\frac12}}\subset\mf g_{\frac12}$
be the dual basis with respect to the nondegenerate skewsymmetric pairing 
$(f|[\cdot\,,\,\cdot])$ on $\mf g_{\frac12}$.
Equivalently, letting $u_k=[f,v_k]$, we have that 
$\{u_k\}_{k\in J_{-\frac12}}\subset\mf g_{-\frac12}$ 
and $\{v^k\}_{k\in J_{-\frac12}}\subset\mf g_{\frac12}$
are dual bases with respect to $(\cdot\,|\,\cdot)$.

An explicit description of the classical $\mc W$-algebra $\mc W=\mc W_z(\mf g,f)$,
associated to the Lie algebra $\mf g$ and the minimal nilpotent element $f$,
is as follows.
As a differential algebra, it is $\mc W=S(\mb F[\partial]\mf g^f)$,
namely the algebra of differential polynomials
in the differential variables $\{a_i\}_{i\in J_0^f}\subset\mf g_0^f$, 
$\{u_k\}_{k\in J_{\frac12}}\subset\mf g_{-\frac12}$, and $f$.
We also let $L=f+\frac{1}{2}\sum_{i\in J_0^f}a_i^2\in\mc W$
(which we can take as a differential generator in place of $f$).
The $\lambda$-brackets on generators are given by Table
\ref{table1} ($a,b\in\mf g_0^f$, $u,u_1\in\mf g_{-\frac12}$).

\begin {table}[H]
\caption{$\lambda$-brackets among generators of $\mc W$ for minimal nilpotent $f$} \label{table1} 
\begin{center}
\begin{tabular}{c||c|c|c}
\phantom{$\Bigg($} 
$\{\cdot\,_\lambda\,\cdot\}_{z}$
& $L$ & $b$ & $u_1$ \\
\hline
\hline \phantom{$\Bigg($} 
$L$ & $\begin{array}{l} (\partial+2\lambda)L \\ -(x| x)\lambda^3+4(x|x)z\lambda \end{array}$
& $(\partial+\lambda)b$ & $\big(\partial+\frac32\lambda\big)u_1$ \\
\hline \phantom{$\Bigg($}
$a$ & $\lambda a$ & $[a,b]+(a|b)\lambda$ & $[a,u_1]$ \\
\hline \phantom{$\Bigg($}
$u$ & $\big(\frac12\partial+\frac32\lambda\big)u$ & $[u,b]$ & eq.\eqref{20130320:eq2} 
\end{tabular}
\end{center}
\end{table}

\noindent
The $\lambda$-bracket of two elements $u$ and $u_1$ of $\mf g_{-\frac12}$is
\begin{equation}\label{20130320:eq2}
\begin{array}{c}
\displaystyle{
\{u_\lambda u_1\}_{z}=
\sum_{k\in J_{-\frac12}}[u,v^k]^\sharp[u_1,v_k]^\sharp
+\big(\partial+2\lambda\big)[u,[e,u_1]]^\sharp
} \\
\displaystyle{
+\frac{(e|[u,u_1])}{2(x|x)}f
-\lambda^2(e|[u,u_1])+z(e|[u,u_1])\,.
}
\end{array}
\end{equation}

Associated to this $1$-parameter family of PVA $\lambda$-brackets
we have compatible Poisson structures 
$H_0$ and $H_1$,
defined by 
$\{\cdot\,_\lambda\,\cdot\}_{z}=\{\cdot\,_\lambda\,\cdot\}_{H_1}-z\{\cdot\,_\lambda\,\cdot\}_{H_0}$.
It follows by \cite[Theorem 4.18]{DSKV13a}
that we can find an infinite sequence of linearly independent local
functionals $\tint h_n\in\mc W/\partial\mc W$, $n\in\mb Z_+$, 
starting with $h_0=f$,
satisfying the Lenard-Magri recursive relations \eqref{20130604:eq7}.
The first few integrals of motion and the corresponding equations of the  
integrable hierarchy are computed in \cite[Section 6]{DSKV13b}.


Note that, by Table \ref{table1}, all elements of $\mf g_0^f$ are central
for the Poisson structure $H_0$. 
Therefore, by Theorem \ref{20140117:thm} we have a Dirac modified bi-Poisson structure
$(H_0,H_1^D)$ on $\mc W$
with respect to the constraints $\{a_i\}_{i\in J_0^f}$,
and the induced Dirac reduced bi-Poisson structure 
$(\overline{H}_0,\overline{H}_1^D)$ on $\overline{\mc W}=\mc W/\langle a_i\rangle_{i\in J_0^f}$.
Let $\ell=|J_0^f|$ and $N=|J_0^f|+|J_{-\frac12}|+1$.
By Definition \ref{20130529:def}
we have $H_1^D=H_1+BC^{-1}B^*$, where
$C(\partial)\in\Mat_{\ell\times\ell}\mc W[\partial]$ 
and 
$B(\partial)\in\Mat_{N\times\ell}\mc W[\partial]$
are matrix differential operators with entries as follows
($i,j\in J_0^f$, $k\in J_{-\frac12}$):
\begin{equation}\label{20140121:eq1}
\left\{
\begin{array}{l}
B_{ij}(\partial)=C_{ij}(\partial)
=[a_j,a_i]+\delta_{ij}\partial\,\\
B_{kj}(\partial)
=[a_j,u_k]\,,\\
B_{Nj}(\partial)
=a_j\partial\,.
\end{array}
\right.
\end{equation}
\begin{proposition}\label{20140123:prop}
The Lenard-Magri recursive relations \eqref{20131203:eq1} hold for the bi-Poisson structure
$(H_1^D,H_0)$.
Consequently, we get an induced integrable bi-Hamiltonian hierarchy in $\overline{\mc W}$.
\end{proposition}
\begin{proof}
For the first assertion,
by Proposition \ref{20131203:prop2} it suffices to show that $\ker B(\partial)=0$ in
$\widetilde{\mc K}^\ell$
(recall that $\widetilde{\mc K}$ denotes the linear closure of $\mc K$,
and $\mc K$ is the field of fractions of $\mc W$).
Let $F=(F_i)_{i=1}^\ell\in\widetilde{\mc K}$ be an element of the kernel of $B(\partial)$.
Looking at the first $\ell$ rows of the matrix $B(\partial)$,
we get the equations
\begin{equation}\label{20140122:eq1}
F_i'=\sum_{j\in J_0^f}[a_i,a_j]F_j
\,\,,\,\,\,\,
i\in J_0^f
\,.
\end{equation}
Let $\mc W_0=\mb F[a_i^{(n)}\,|\,i\in J_0^f,n\in\mb Z_+]\subset\mc W$,
be the algebra of differential polynomials in the differential variables $\{a_i\}_{i\in J_0^f}$,
let $\mc K_0$ be its differential field of fractions,
and let $\widetilde{\mc K}_0$ be its linear closure.
It is a differential subfield of $\widetilde{\mc K}$
with the same subfield of constants $\overline{\mb F}$.
The space of solutions of equation \eqref{20140122:eq1} in $\widetilde{\mc K}^\ell$
is an $\overline{\mb F}$-linear subspace of $\widetilde{\mc K}_0^\ell$
of dimension $\ell$.
Let $E\in\Mat_{\ell\times\ell}\widetilde{\mc K}_0$
be a non-degenerate matrix,
whose columns form a basis of the space of solutions of equation \eqref{20140122:eq1}.
Then, all solutions of equation \eqref{20140122:eq1} have the form
\begin{equation}\label{20140122:eq3}
F=EC\,,
\end{equation}
for some constant vector $C\in\overline{\mb F}^\ell$, 
see e.g. \cite{CDSK13b}.

Next, consider the following $N-\ell-1$ rows of the matrix $B(\partial)$.
We get the equations
\begin{equation}\label{20140122:eq4}
\sum_{i\in J_0^f,k\in J_{-\frac12}}[u_k,a_i]E_{ij}C_j=0
\,\,,\,\,\,\,
k\in J_{-\frac12}
\,.
\end{equation}
The left hand side of equation \eqref{20140122:eq4} lies in 
$\mf g_{-\frac12}\otimes\widetilde{\mc K}_0$.
Hence, we can apply $\xi=(v|\,\cdot)\in\mf g_{\frac12}^*$ to it
(considered as a linear map $\xi:\,\mf g_{-\frac12}\otimes\widetilde{\mc K}_0\to\widetilde{\mc K}_0$):
\begin{equation}\label{20140122:eq5}
\sum_{i\in J_0^f,k\in J_{-\frac12}}
([v,u_k]|a_i)E_{ij}C_j=0
\,,
\end{equation}
for all $v\in\mf g_{\frac12}$ and $k\in J_{-\frac12}$.
Note that $[\mf g_{\frac12},\mf g_{-\frac12}]=\mf g_0$
(it follows by the fact that $\mf g$ is simple and
$\mb Ff\oplus\mf g_{-\frac12}\oplus[\mf g_{\frac12},\mf g_{-\frac12}]\oplus\mf g_{\frac12}\oplus\mb Fe$
is clearly an ideal of $\mf g$).
Since the inner product $(\cdot\,|\,\cdot)$ is non-degenerate on $\mf g_0^f$, 
equation \eqref{20140122:eq5} implies $EC=0$,
from which we get that $C=0$, by the non-degeneracy of the matrix $E$.
Hence $F=0$, as required.

For the last assertion, 
we consider the images of the conserved densities $h_n$ in $\overline{\mc W}$,
and the corresponding variational derivatives 
$\frac{\delta \overline{h}_n}{\delta u}\in\overline{\mc W}^{N-\ell}$.
It follows by \cite[Lem.4.15]{DSKV13a} that they span an infinite dimensional space.
Since $\overline{H}_0$ is a non-degenerate matrix differential operator,
it follows that the images of the elements $P_n$ in $\overline{\mc W}^{N-\ell}$
span an infinite dimensional space as well.
\end{proof}
The Dirac reduced Poisson structures are explicitly as follows
($h,k\in J_{-\frac12}$):
$$
\left\{\begin{array}{l}
\displaystyle{
(\overline{H}_0)_{hk}(\partial)=(e|[u_h,u_k])
}\\
\displaystyle{
(\overline{H}_0)_{NN}(\partial)=-4(x|x)\partial
}\\
\displaystyle{
(\overline{H}_0)_{Nk}(\partial)=(\overline{H}_0)_{kN}(\partial)=0
}
\end{array}\right.
\,,
$$
and
$$
\left\{\begin{array}{l}
\displaystyle{
(\overline{H}_1^D)_{hk}(\partial)
=
\sum_{i\in J_{0}^f}[a_i,u_h]\partial^{-1}\circ[a_i,u_k]
-\frac{(e|[u_h,u_k])}{2(x|x)}f
+(e|[u_h,u_k])\partial^2
}\\
\displaystyle{
(\overline{H}_1^D)_{NN}(\partial)
=
f'+2f\partial-(x| x)\partial^3
}\\
\displaystyle{
(\overline{H}_1^D)_{kN}(\partial)
=
u_k'+\frac32u_k\partial
}\\
\displaystyle{
(\overline{H}_1^D)_{Nk}(\partial)
=
\frac12u_k'+\frac32u_k\partial
}
\end{array}\right.
\,.
$$
The first two conserved densities are
$$
h_0=f
\,\,\,\,
\text{ and }
\,\,\,\, 
h_1=-\frac1{8(x|x)}f^2
-\frac12\sum_{k\in J_{-\frac12}} [f,v^k]u_k^\prime\,,
$$
and the first two equations of the reduced bi-Hamiltonian hierarchy are 
(for $u\in\mf g_{-\frac12}$)
$\frac{du}{dt_0}=u'$,
$\frac{d f}{dt_0}=f'$, and
(cf. \cite[eq.(6.21)]{DSKV13b})
$$
\begin{array}{rcl}
\displaystyle{
\frac{du}{dt_1}
}
&=&
\displaystyle{
u'''
-\frac3{4(x|x)}fu'
-\frac3{8(x|x)}uf'
-\frac12\sum_{i\in J_0^f,k\in J_{-\frac12}}
[a_i,u][a^i,[f,v^k]][f,v_k]
\,,
}\\
\displaystyle{
\frac{d f}{dt_1}
}
&=&
\displaystyle{
\!\!\frac14f'''
-\frac3{4(x|x)}ff'
+\frac32\!\sum_{k\in J_{-\frac12}}\!u_k[f,v^k]''
\,.}
\end{array}
$$

In the case of $\mf g=\mf{sl}_n$, $n\geq3$ (and only in this case) there is a
(unique up to a constant factor) non-zero element $c$ in the center of $\mf g_0^f$.
Then we have additional conserved densities $h_{\tilde n}$, $n\in\mb Z_+$, of which the
first two are
$$
h_{\tilde 0}=c
\,\,\,\,
\text{ and }
\,\,\,\,
h_{\tilde 1}
=\sum_{k\in J_{-\frac12}}u_k[c,[f,v^k]]
\,.
$$
The corresponding first two reduced bi-Hamiltonian equations are
(for $u\in\mf g_{-\frac12}$)
$\frac{du}{dt_{\tilde0}}=[c,u]$,
$\frac{d f}{dt_{\tilde0}}=0$, and
(cf. \cite[eq.(6.22)]{DSKV13b})
$$
\begin{array}{rcl}
\displaystyle{
\frac{du}{dt_{\tilde1}}
}
&=&
\displaystyle{
[c,u]''
-\frac1{2(x|x)}f[c,u]
\,,
}\\
\displaystyle{
\vphantom{\bigg(}
\frac{d f}{dt_{\tilde1}}
}
&=&
\displaystyle{
\sum_{k\in J_{-\frac12}}(u_k[c,[f,v^k]])^\prime
\,.}
\end{array}
$$
Due to the observations at the end of Section \ref{sec:2.4}
all the flows $\frac{d}{dt_n}$ and $\frac{d}{dt_{\tilde m}}$
commute and the local functionals $\tint h_n$, $\tint h_{\tilde m}$
are in involution, for $n,m\in\mb Z_+$.

\section{Dirac reduced short DS hierarchy}
\label{sec:4.4}

Let $\mf g$ be a simple Lie algebra with 
a non-degenerate symmetric invariant bilinear form $(\cdot\,|\,\cdot)$.
Recall that, by definition, for a short nilpotent element $f\in\mf g$,
and an $\mf{sl}_2$-triple $\{f,h=2x,e\}$, we have the 
$\ad x$-eigenspace decomposition $\mf g=\mf g_{-1}\oplus\mf g_{0}\oplus\mf g_{1}$.
Moreover, we have the orthogonal decomposition
$\mf g_0=\mf g_0^f\oplus[f,\mf g_1]$,
and $\mf g_0^f=\mf g_0^e$, $[f,\mf g_1]=[e,\mf g_{-1}]$.
Let ${}^\sharp:\,\mf g_0\to\mf g_0^f$
be the corresponding orthogonal projection.
Let $\{a_i\}_{i\in J_0^f}\subset\mf g_0^f$ be an orthonormal basis of $\mf g_0^f$.
Let also $\{u_k\}_{k\in J_{-1}}\subset\mf g_{-1}$ be a basis of
$\mf g_{-1}$ and let $\{u^k\}_{k\in J_{-1}}\subset\mf g_1$
be the dual basis with respect to $(\cdot\,|\cdot)$.

The classical $\mc W$-algebra $\mc W=\mc W(\mf g,f)$ 
is, as differential algebra,
the algebra of differential polynomials in the differential variables
$\{a_i\}_{i\in J_0^f}\subset\mf g_0^f$ and $\{u_k\}_{k\in J_{-1}}\subset\mf g_{-1}$.
The $\lambda$-bracket on generators are given by Table \ref{table3} 
($a,b\in\mf g_0^f$, $u,u_1\in\mf g_{-1}$). 

\begin {table}[h]
\caption{$\lambda$-brackets among generators of $\mc W$ for short nilpotent $f$} \label{table3}
\begin{center}
\begin{tabular}{c||c|c}
\phantom{$\Bigg($} 
$\{\cdot\,_\lambda\,\cdot\}_{z}$
& $b$ & $u_1$ \\
\hline
\hline \phantom{$\Bigg($} 
$a$ & $[a,b]+(a|b)\lambda$ & $[a,u_1]$  \\
\hline \phantom{$\Bigg($}
$u$ & $[u,b]$ & eq.\eqref{20130320:eq2b}
\end{tabular}
\end{center}
\end{table}

The $\lambda$-bracket of $u,u_1\in\mf g_{-1}$ is
\begin{equation}\label{20130320:eq2b}
\begin{array}{l}
\displaystyle{
\{u_\lambda u_1\}_{z}
=\frac12\sum_{k\in J_1}(u\circ u_k)[u_1,u^k]^\sharp
-\frac12\sum_{k\in J_1}(u_1\circ u_k)[u,u^k]^\sharp
}\\
\displaystyle{
+\frac14\sum_{h,k\in J_1}[[e,u_h],[e,u_k]]
[u,u^h]^\sharp[u_1,u^k]^\sharp
-\frac12(\partial+2\lambda)(u\circ u_1)
}\\
\displaystyle{
+\frac14(\partial+2\lambda)\sum_{k\in J_1}[[e,u],[e,u_k]][u_1,u^k]^\sharp
+\frac14\sum_{k\in J_1} [[e,u_1],[e,u_k]] (\partial+\lambda) [u,u^k]^\sharp
}\\
\displaystyle{
-\frac14\left(3\lambda^2+3\lambda\partial+\partial^2\right)[[e,u],[e,u_1]]
+\frac14(e|u\circ u_1)\lambda^3
}\\
\displaystyle{
+z[[e,u],[e,u_1]]^\sharp
-(e|u\circ u_1)z\lambda
\,,
}
\end{array}
\end{equation}
where $u\circ u_1=[[e,u],u_1]$, for all $u,u_1\in\mf g_{-1}$.

Associated to this $1$-parameter family of PVA $\lambda$-brackets
we have compatible Poisson structures  $H_0$ and $H_1$,
defined by 
$\{\cdot\,_\lambda\,\cdot\}_{z}=\{\cdot\,_\lambda\,\cdot\}_{H_1}-z\{\cdot\,_\lambda\,\cdot\}_{H_0}$.
It follows by \cite[Theorem 4.18]{DSKV13a}
that we can find an infinite sequence of linearly independent local
functionals $\tint h_n\in\mc W/\partial\mc W$, $n\in\mb Z_+$, 
starting with $h_0=f$,
satisfying the Lenard-Magri recursive relations \eqref{20130604:eq7}.
The first few integrals of motion and the corresponding equations of the  
integrable hierarchy are computed in \cite[Section 7]{DSKV13b}.

Note that, by Table \ref{table1}, all elements of $\mf g_0^f$ are central
for the Poisson structure $H_0$. 
Therefore, by Theorem \ref{20140117:thm} we have a Dirac modified bi-Poisson structure
$(H_0,H_1^D)$ on $\mc W$
with respect to the constraints $\{a_i\}_{i\in J_0^f}$,
and the induced Dirac reduced bi-Poisson structure 
$(\overline{H}_0,\overline{H}_1^D)$ on $\overline{\mc W}=\mc W/\langle a_i\rangle_{i\in J_0^f}$.
Let $\ell=|J_0^f|$ and $N=|J_0^f|+|J_{-1}|$.
By Definition \ref{20130529:def}
we have $H_1^D=H_1+BC^{-1}B^*$, where
$C(\partial)\in\Mat_{\ell\times\ell}\mc W[\partial]$ 
and 
$B(\partial)\in\Mat_{N\times\ell}\mc W[\partial]$
are matrix differential operators with entries as follows
($i,j\in J_0^f$, $k\in J_{-1}$):
\begin{equation}\label{20140121:eq1b}
\left\{
\begin{array}{l}
B_{ij}(\partial)=C_{ij}(\partial)
=[a_j,a_i]+\delta_{ij}\partial\,\\
B_{kj}(\partial)
=[a_j,u_k]\,.
\end{array}
\right.
\end{equation}
\begin{proposition}\label{20140123:propb}
The Lenard-Magri recursive relations \eqref{20131203:eq1} hold for the bi-Poisson structure
$(H_1^D,H_0)$.
Consequently, we get an induced integrable bi-Hamiltonian hierarchy in $\overline{\mc W}$.
\end{proposition}
\begin{proof}
It is along the same lines as the proof of Proposition \ref{20140123:prop}.
\end{proof}
The Dirac reduced Poisson structures are explicitly as follows
($h,k\in J_{-1}$):
$$
(\overline{H}_0)_{hk}(\partial)
=
(e|u_h\circ u_k)\partial
\,,
$$
and
$$
(\overline{H}_1^D)_{hk}(\partial)
=
\sum_{i\in J_{0}^f}[a_i,u_h]\partial^{-1}\circ[a_i,u_k]
-\frac12(u_h\circ u_k)^\prime-(u_h\circ u_k)\partial
+\frac14(e|u_h\circ u_k)\partial^3
\,.
$$
The first two equations of the reduced bi-Hamiltonian hierarchy are
(for $u\in\mf g_{-1}$)
$\frac{du}{dt_0}=u'$, and
$$
\frac{du}{dt_1}=
\frac14u'''
+\frac34\sum_{h,k\in J_{-1}}
(u^k*u^h
|u)
u_h u_k'
\,,
$$
where $*$ is the Jordan product on $\mf g_1$ defined 
by $a*b=[[f,a],b]$, for every $a,b\in\mf g_1$.
The last equation is, after a rescaling of the variables, the Svinolupov equation
associated to this Jordan product, \cite{Svi91}.
We thus provided a bi-Hamiltonian structure for such equation
and we proved its integrability.


%
\end{document}